\documentclass{amsart}
\usepackage{amssymb}
\usepackage{mathrsfs}

\newtheorem{theorem}{Theorem}[section]
\newtheorem{proposition}[theorem]{Proposition}
\newtheorem{lemma}[theorem]{Lemma}
\newtheorem{definition}[theorem]{Definition}
\newtheorem{corollary}[theorem]{Corollary}
\newtheorem{assumption}[theorem]{Assumption}
\newtheorem{remark}[theorem]{Remark}

\numberwithin{equation}{section}

\newcommand{\Om}{\Omega}

\newcommand{\ep}{\varepsilon}
\renewcommand{\phi}{\varphi}
\newcommand{\de}{\delta}
\newcommand{\la}{\lambda}

\renewcommand{\P}{\mathbb P}
\newcommand{\Q}{\mathbb Q}
\newcommand{\E}{\mathbb E}
\newcommand{\R}{\mathbb R}
\newcommand{\N}{\mathbb N}

\newcommand{\con}{\triangleleft}
\newcommand{\qn}{(Q^n)}
\newcommand{\pn}{(P^n)}
\newcommand{\ind}{1\!\kern-1pt \mathrm{I}}
\newcommand{\rsto}{]\!\kern-1.8pt ]}
\newcommand{\lsto}{[\!\kern-1.7pt [}

\begin{document}


\title[Asymptotic arbitrage with small transaction costs]{Large financial markets and asymptotic arbitrage with small transaction costs}


\author{Irene Klein}
\address{Department of Statistics and Operations Research, University of Vienna, Br\"{u}nnerstrasse 72, 1210 Vienna, Austria.}
\email{irene.klein@univie.ac.at}

\author{Emmanuel Lepinette}
\address{Ceremade, Universit\'e Paris Dauphine, Place du Mar\'echal De Lattre De Tassigny, 75775 Paris cedex 16, France}
\email{emmanuel.lepinette@ceremade.dauphine.fr}

\author{Lavinia Ostafe}
\address{Faculty of Mathematics, University of Vienna, Nordbergstrasse 15, 1090 Vienna, Austria.}
\email{lavinia.ostafe@univie.ac.at}
\thanks{The third author gratefully acknowledges financial support from the Austrian
Science Fund (FWF) under grant P19456 and from the European Research
Council (ERC) under grant No.~247033.}

\date{\today}%

\begin{abstract}
We give  characterizations of asymptotic arbitrage of the first and second kind and
of strong asymptotic arbitrage for large financial markets with small proportional transaction costs $\la_n$ on market $n$
in terms of contiguity properties of sequences of equivalent probability measures induced by $\la_n$--consistent price systems.
 These results are analogous to the frictionless case, compare \cite{Kab:Kra:1998}, \cite{K:S:1996}. Our setting is simple, each market $n$ contains two assets with continuous price processes. The proofs use quantitative versions of the Halmos--Savage Theorem, see \cite{K:Sch:1996}, and a monotone convergence result of nonnegative local martingales. Moreover, we present an example admitting a strong asymptotic arbitrage without transaction costs; but with transaction costs $\la_n>0$ on  market $n$ ($\la_n\to0$ not too fast) there does not exist any form of asymptotic arbitrage.
\end{abstract}

\maketitle

\section{Introduction}\label{intro}

In the classical theory of mathematical finance a crucial role is played by the notion of \textit{arbitrage}, which is the cornerstone of the
option pricing theory that goes back to F.~Black, R.~Merton and M.~Scholes \cite{Bl:Sch}. In the past decades, significant work has been done to
develop this theory for models with a finite number of assets (to which we will refer to as ``small'' market models).
As the ``real world'' financial market can contain a very large (even unbounded) number of traded securities the natural idea of approximating  such a ``large'' market by a
sequence of small models came up. Even if each of the small market models is arbitrage--free, by investing in a large
enough number of them one may obtain an asymptotic form of arbitrage in the limit. So, it is clear that in a large financial market one has to exclude asymptotic arbitrage opportunities as well.

Starting from the ideas of Ross and Huberman, \cite{Ross} and \cite{Hu}, to describe a financial market by a
sequence of market models with a finite number of securities, Kabanov and Kramkov \cite{Kab:Kra:1994}
introduced the concept of a large
financial market that fits the continuous--time framework. They consider a sequence of price
processes $\{(S_t^n)_{t\in\R_+}\}_{n\in\N}$ based on a sequence of filtered probability spaces rather than a single stochastic
process $(S_t)_{t\in\R_+}$ based on a fixed filtered probability space.
Kabanov and Kramkov introduced the notion of \emph{asymptotic
arbitrage} and distinguished between two kinds: \emph{asymptotic arbitrage of the first kind} (AA1) and
\emph{asymptotic arbitrage of the second kind} (AA2). AA1 can be seen as the opportunity of getting arbitrarily rich with strictly positive
probability by taking an arbitrarily small risk. AA2 gives the opportunity of gaining at least something (which could be a small amount) with
probability arbitrarily close to one, while taking the risk of losing an amount of money which is uniformly bounded.

Extensive work has been done in the case of a small market model without transaction costs to relate the absence of arbitrage opportunities to the existence
of equivalent martingale measures, see e.g. \cite{H:K}, \cite{H:P}, \cite{Kr}, \cite{D:M:W}, \cite{D:S:94}, \cite{D:S:98}, \cite{K:K} and \cite{K:S}.
A similar situation occurs also in the context of a large financial market. Under the assumption that each small market model does not allow for any kind of arbitrage opportunity one gets a connection
between the absence of asymptotic arbitrage and some properties of the sequences of equivalent martingale measures $Q^n$ for $S^n$. The first results in this direction were obtained by
Kabanov and Kramkov under the assumption that each small market model is complete. In \cite{Kab:Kra:1994}
they established necessary and sufficient conditions for the absence of asymptotic arbitrage of the first and of the second kind in terms of contiguity properties of the sequence of equivalent martingale measures with respect to the sequence of objective measures. As contiguity of sequences of measures is the generalization of absolute continuity of measures these theorems can be considered as versions of the Fundamental Theorem of Asset Pricing for large financial markets. In \cite{K:S:1996} and \cite{K:Sch:1996}, Klein and Schachermayer and
in  \cite{Kab:Kra:1998} Kabanov and Kramkov extended the theorems of  \cite{Kab:Kra:1994} to the incomplete market case.

All the above--mentioned results were obtained for a large financial market without transaction costs. The main
goal of this paper is to study the asymptotic arbitrage opportunities when
each small market model of index $n$ is subject to transaction costs $\la_n$. When one introduces transaction
costs, the usual notion of an equivalent martingale measure that is used in a market without friction, is replaced
by the concept of a $\la$--consistent price system ($\la$--CPS), see e.g.  \cite{C:S}, \cite{G:R:S:2008},  \cite{Ge:MK:Sch}, \cite{K:Saf}. In analogy with the classical case of a small market without
transaction costs, it is here possible to relate the absence of arbitrage opportunities to the existence of
$\la$--CPSes. In this paper  we will give  similar criteria for the case of  a large financial market. To be precise we give characterizations of the absence of the various asymptotic arbitrage
opportunities in terms of contiguity properties of the sequence of sets of measures $Q^n$ induced by $\la_n$--CPSes $(Q^n,\widetilde{S}^n)$
with respect to the sequence of objective measures $(P^n)$.

We will now briefly describe the structure of the present paper in detail.
We concentrate on a simple setting for each small market model with proportional transaction costs $\la_n$. That is, we assume that
each market $n$ is given by two assets, a
risk--free and a risky one with continuous paths. We impose transaction costs $\la_n$ in the following way: if we buy one unit of the risky asset at time $t$ we have to pay $S^n_t$, if we sell one unit at time $t$ we receive $(1-\la_n)S^n_t$.
In Section~2 we introduce the market model. In Section~3 we give the definitions of asymptotic arbitrage in the presence of small transaction costs $\la_n$ on market $n$. Moreover, we present our main results, which turn out to be completely analogous  to the frictionless case, compare \cite{K:S:1996}, \cite{Kab:Kra:1998}. Indeed, the first theorem gives  that the absence of asymptotic arbitrage of the first kind with transaction costs $\la_n$ on market $n$ is equivalent to the existence of a sequence of $\la_n$--CPS $(Q^n,\widetilde{S}^n)$ such that the sequence of objective measures $(P^n)$ is contiguous with respect to the sequence $(Q^n)$. The second theorem gives that the absence of asymptotic arbitrage of second kind with transaction costs $\la_n$ on market $n$ is equivalent to weak contiguity of the sequence of sets of measures $Q^n$ coming from a $\la_n$--CPS $(Q^n,\widetilde{S}^n)$   with respect to the sequence of objective measures $P^n$. Moreover, in a third theorem we characterize the notion of
strong asymptotic arbitrage in terms of a condition of entire separability of the sequence of the sets of measures $Q^n$ and the sequence of objective measures $(P^n)$.
In Section~4 we give the proofs of the main results. One tool for these proofs is a monotone convergence result for a sequence of nonnegative local martingales which we present in an extra section, Section~5, as it seems interesting in itself.
In Section~6 we present an example of a large financial market that shows that the introduction of transaction costs influences the asymptotic arbitrage opportunities. We present a large financial market that allows for a strong asymptotic arbitrage. However, if we impose transaction costs (however small) the asymptotic arbitrage disappears. We  can even show that, if we let the transaction costs $\la_n$ on market $n$ converge to 0 not too fast, then there still does not exist any form of asymptotic arbitrage. We show this by directly constructing a $\la_n$--CPS with a unique equivalent martingale measure $Q^n$ such that the sequence $(Q^n)$ is contiguous with respect to $(P^n)$ and vice versa.

At the end of this introduction  a  comparison of our results with the similar recent results of Lepinette and Ostafe \cite{Le:Lav} is in order. Let us first mention that the setting of two assets is more natural  if one is interested in the behavior when the transaction costs $\la_n$ converge to 0 for $n\to\infty$. Section~\ref{ex} is one step in this direction.
The more general high dimensional setting as in  \cite{Le:Lav} is based on the  framework of Kabanov's modelling of multi--dimensional currency markets in a num\'eraire--free way, see e.~g.~\cite{K:L}, \cite{K:S}, \cite{K:R:S}.
The authors of \cite{Le:Lav}  can extend the results of \cite{Kab:Kra:1998}, \cite{K:S:1996} to multidimensional large financial markets with transaction costs.
We still think that it is interesting to prove the results in the two-dimensional case directly. On one hand our proofs are based on rather direct applications of the quantitative versions of the Halmos--Savage--Theorem of \cite{K:Sch:1996} and we do not have to go into the involved details of the general cone--setting. On the other hand, although
our setting clearly is a special case of the one provided in \cite{Le:Lav}, there is a substantial difference in the way the concepts of asymptotic arbitrage are defined and therefore the theorems are similar but not equivalent.
Usually
in mathematical finance, in an arbitrage condition, one considers trading strategies with value processes that are bounded from below by the same constant at all
times $0\leq t\leq T$. In this way phenomena like doubling strategies are excluded. Note that, in the asymptotic arbitrage definitions of the
present paper we require the value processes to be bounded from below by the same constant for all $n$ and all times $0\leq t\leq T_n$, whereas this is not the case in the asymptotic arbitrage definitions of \cite{Le:Lav}.
They only assume the same bound from below for all $n$ at the terminal times $T_n$. At times $t<T_n$ the bound from below could depend on $n$. This means, in their versions of asymptotic arbitrage, when $n$ goes to infinity, one could make arbitrarily
big losses at times $t<T_n$ before reaching the asymptotic arbitrage at terminal time $T_n$, which could be compared to the phenomenon of doubling strategies. Moreover, this is not in complete analogy with the original definitions of asymptotic arbitrage opportunities as appearing in \cite{K:S:1996}, \cite{Kab:Kra:1998}, where the uniform bound from below for all $n$ and $t$ was crucial.
Therefore, we consider
our present definitions with the same lower bound for all $n$ at all time points $t$ as the more natural
versions of asymptotic arbitrage.
In the no transaction costs case, however, this difference does not even occur as  having the bound from below at maturity
$T$ implies the same bound from below for all times $t<T$ (as we assume NA for each market $n$), see \cite{D:S:94}. In the transaction
costs setting, the situation is more subtle. A not yet published example of W.~Schachermayer \cite{Sch} shows that,  even if
one excludes arbitrage with transaction costs $\la$, it is possible to construct an example in which the value
process of an admissible trading strategy at the terminal time is bounded below by $-1$ but it is strictly smaller than $-1$ with positive
probability at a prior time point. A result of
W.~Schachermayer \cite{Sch}, which we state in the Appendix, see Proposition~\ref{schach}, provides us with the appropriate tool for our proofs.
Under the assumption that we have NA {\it for all} $\la>0$, this result gives us the same
bound from below at all times $t\leq T$ if the value process is bounded from below at the terminal time $T$ (this is
 then in analogy to the no transaction costs case). In order to apply this we  impose that
on each market $n$ there does not exist an arbitrage opportunity with any transaction costs $\la>0$ (however small).
This is in
the spirit of large financial markets where there is always assumed that each small market is arbitrage-free in
every possible sense. As our no
asymptotic arbitrage conditions are weaker than those in \cite{Le:Lav}, our theorems are stronger in the following aspect: the more involved implications,
which are that no asymptotic arbitrage implies the corresponding contiguity property, are  stronger.

\section{Large financial markets with proportional transaction costs}\label{model}

A large financial market consists of a sequence of  market models, which are given in the following way. Let
$(\Omega^n,\mathcal{F}^n, (\mathcal{F}^n_t)_{0\leq t\leq T_n}, P^n)$, $n\in\N$, be a sequence of filtered probability spaces where the filtration
satisfies the usual assumptions. For each $n$ we are given a risk--free asset $B^n$ normalized to $B^n_t=1$ and a risky asset $S^n$, where $(S^n_t)_{0\leq t\leq T_n}$ is adapted  to $(\mathcal{F}^n_t)_{0\leq t\leq T_n}$ with continuous and strictly positive paths. Moreover, in order to be able to apply the crucial Proposition~\ref{schach} which appears in an unpublished work of W.~Schachermayer,
we use the same assumptions: that is, we assume that the filtration $(\mathcal{F}^n_t)_{0\leq t\leq T_n}$ is generated by a $d(n)$--dimensional Brownian motion $(W^n_t)_{0\leq t\leq T_n}$. (All local martingales with respect to an equivalent measure will have continuous paths).

We assume that on each market $n$ we have to pay proportional transaction costs $\la_n$. As usual in the context of large financial markets we will assume
an appropriate condition on each market $n$ that will exclude any form of arbitrage there. In our case we want to guarantee that
on each market $n$ there are no arbitrage opportunities with small transaction costs.
To this end we recall the notion of a consistent pricing system as used, for example, in \cite{Ge:MK:Sch}, which replaces the notion of an equivalent martingale measure that is used in a market without transaction costs.

\begin{definition}\label{CPS}
A $\la$--consistent price system  (CPS) for the market $n$ is a pair $({Q}^n,\widetilde{S}^n)$  of a probability measure
${Q}^n\sim P^n$ and a process $(\widetilde{S}^n_t)_{0\leq t\leq T^n}$ which is a local martingale under ${Q}^n$ such that
\begin{equation}
 (1-\la)S^n_t\leq \widetilde{S}^n_t\leq S^n_t,\ \text{a.s., for all $t\in[0,T^n]$}.
\end{equation}
Denote by $\mathcal{M}^n(\la)$ the set of all probability measures $Q\sim P^n$ inducing a $\la$--CPS, that is
\begin{equation}\label{CPSset}
\mathcal{M}^n(\la)=\{Q\sim P^n:\text{$\exists$  a $Q$--local martingale $\widetilde{S}^n$ with $(Q,\widetilde{S}^n)$ $\la$--CPS}\}
\end{equation}
\end{definition}

The following assumption will be used throughout the paper. For each market $n$ this condition is equivalent to
the absence of arbitrage with arbitrarily small transaction costs $\la>0$ on each market $n$, see \cite{G:R:S:2008}
(we state the theorem in the appendix, Theorem~\ref{NA-L}).

\begin{assumption}\label{NA}
$\mathcal{M}^n(\la)\neq\emptyset$ for all $\la>0$ and all $n\in\N$.
\end{assumption}

In order  to give a meaning to a no arbitrage condition in the presence of transaction costs $\la$ we define the
following trading strategies. The definition goes back to \cite{Kal:MK:2010}, \cite{Kal:MK:2011} and
\cite{Ge:MK:Sch}. Here everything works in the same way for each market $n$, therefore we
omit the superscript $n$ for the moment.

\begin{definition}
A self--financing trading strategy with zero endowment is a pair of right--continuous, adapted finite--variation processes
$(\phi^0_t,\phi^1_t)_{0\leq t\leq T}$ such that
\begin{enumerate}
\item $\phi_{0-}^0=\phi_{0-}^1=0$
\item $\phi_t^0=\phi_t^{0,\uparrow}- \phi_t^{0,\downarrow}$ and $\phi_t^1=\phi_t^{1,\uparrow}- \phi_t^{1,\downarrow}$, where
$\phi_t^{0,\uparrow}$, $\phi_t^{0,\downarrow}$, $\phi_t^{1,\uparrow}$, $\phi_t^{1,\downarrow}$ are the decompositions of $\phi^0$ and $\phi^1$ into increasing processes starting at $\phi_{0-}^{0,\uparrow}=\phi_{0-}^{0,\downarrow}=\phi_{0-}^{1,\uparrow}=\phi_{0-}^{1,\downarrow}=0$ and satisfying
\begin{equation}\label{selff}
d\phi_t^{0,\uparrow}\leq (1-\la)S_td\phi_t^{1,\downarrow},\quad d\phi_t^{0,\downarrow}\geq S_td\phi_t^{1,\uparrow}, \quad 0\leq t\leq T.
\end{equation}
\item The trading strategy $\phi:=(\phi^0,\phi^1)$ is called admissible if there is $M>0$ such that  the value process (under transaction costs $\la$) $V^{\la,S}_t(\phi)$ satisfies
\begin{equation}\label{valuepr}
V_t^{\la,S}(\phi):=\phi^0_t+(\phi_t^1)^+(1-\la)S_t-(\phi_t^1)^-S_t\geq -M,\text{
a.s. for $0\leq t\leq T$}.
\end{equation}
\end{enumerate}
\end{definition}
Here $\phi^0$ specifies the holdings in the bond and $\phi^1$ the holdings in the stock.
In the above definition, the initial value of the trading strategy is given at the point $0-$. We extend the time
interval $[0,T]$ by the additional point $0-$ in order to be able to cope with a possible jump which may appear
at time 0. We impose the conditions in
\eqref{selff} in order to be solvent, i.e., after liquidating the position in stock, the
resulting amount in bond is non--negative. The positive and negative part of $\phi_t^1$ in \eqref{valuepr} mean
that we liquidate the position in stock, which is selling $\phi_t^1$ units of stock for the price $(1-\la)S_t$
if $\phi_t^1$ is positive, and  buying $-\phi_t^1$ units of stock for the price $S_t$ if $\phi_t^1$ is negative.

Now we present the natural definition of an arbitrage with transaction costs $\la$.

\begin{definition}
$S$ admits arbitrage with transaction costs $\la$ if there is an admissible self-financing trading strategy $\phi=(\phi^0,\phi^1)$ with zero endowment such that $V^{\la,S}_T(\phi)\geq0$ and $P(V^{\la,S}_T(\phi)>0)>0$.
If $S$ does not admit arbitrage with transaction costs $\la$ then we say that the market satisfies the property NA($\la$).
\end{definition}

As stated in the introduction, by Theorem~\ref{NA-L} in the appendix, which goes back to \cite{G:R:S:2008}, Assumption~\ref{NA} means that, on each market $n$, NA($\la$) holds for all $\la>0$.
\section{Asymptotic arbitrage in the presence of small transaction costs}\label{results}

In this chapter we will derive the connection between the absence of asymptotic arbitrage with transaction costs
$\la_n$ and some properties of sequences of $\la_n$--CPS. It turns out that the results are analogous to the
results in the case of no transaction costs, see  \cite{Kab:Kra:1998} and \cite{K:S:1996}. Fix any sequence $(\la_n)$ of  real numbers $0<\la_n<1$.  Let
us now define the notions of asymptotic arbitrage of the first and the second kind and of strong asymptotic arbitrage for
the markets with transaction costs.

\begin{definition}\label{AA1}
 There exists an asymptotic arbitrage of the first kind  with transaction
costs $\la_n$ (AA1($\la_n$)) if there exists a subsequence of markets (again denoted by $n$) and admissible
trading strategies $\phi^n=(\phi^{0,n}, \phi^{1,n})$ with zero endowment for $S^n$ such that
\begin{enumerate}
\item $V^{\la_n,S^n}_t(\phi^n)\geq-c_n$, for all $0\leq t\leq T^n$,
\item $\lim_{n\to\infty}P^n(V_{T^n}^{\la_n,S^n}(\phi^n)\geq C_n)>0$
\end{enumerate}
where $c_n$ and $C_n$ are sequences of positive real numbers with $c_n\to0$ and $C_n\to\infty$.
\end{definition}

\begin{definition}\label{AA2}
 There exists an asymptotic arbitrage of the second kind  with transaction
costs $\la_n$ (AA2($\la_n$))
if there exists a subsequence of markets (again denoted by $n$) and admissible trading strategies $\phi^n=(\phi^{0,n}, \phi^{1,n})$ with zero endowment for $S^n$ and $\alpha>0$ such that
\begin{enumerate}
\item $V^{\la_n,S^n}_t(\phi^n)\geq-1$, for all $0\leq t\leq T^n$,
\item $\lim_{n\to\infty}P^n(V^{\la_n,S^n}_{T^n}(\phi^n)\geq \alpha)=1$.
\end{enumerate}
\end{definition}

\begin{definition}\label{SAA}
 There exists a strong asymptotic arbitrage  with transaction
costs $\la_n$  if there exists a subsequence of markets (again denoted by $n$) and admissible
trading strategies $\phi^n=(\phi^{0,n}, \phi^{1,n})$ with zero endowment for $S^n$ such that
\begin{enumerate}
\item $V^{\la_n,S^n}_t(\phi^n)\geq-c_n$, for all $0\leq t\leq T^n$,
\item $\lim_{n\to\infty}P^n_{T^n}(V^{\la_n,S^n}(\phi^n)\geq C_n)=1$
\end{enumerate}
where $c_n$ and $C_n$ are sequences of positive real numbers with $c_n\to0$ and $C_n\to\infty$.
\end{definition}

As the notion of equivalent probability measure as it appears in the setting of one fixed market has to be
replaced by the notion of contiguity of sequences of  probability  measures, we recall the definition of
contiguity.

\begin{definition}
The sequence of probability measures $\pn$ is contiguous with respect to the sequence of probability measures
$\qn$, $\pn\con\qn$, if the following holds:
for  $A^n\in\mathcal{F}^n$ with $Q^n(A^n)\to0$, for $n\to\infty$, we have that $P^n(A^n)\to0$, for $n\to\infty$.
\end{definition}

Let us stress that, in all our definitions of the various forms of asymptotic arbitrage, we assume that the value processes are bounded from
below by the same constant, for all $t$ and all $n$. This is a crucial difference to the definitions of \cite{Le:Lav} as we already explained in some detail in the introduction.

Now we can formulate the three main results of this paper. The proofs are based on an easy characterization of
the various asymptotic arbitrage properties and follow immediately by quantitative Halmos--Savage results that go
back to \cite{K:Sch:1996} and by a monotone convergence result for local martingales which we prove in Section~\ref{locmart}.

\begin{theorem}\label{1}
Under Assumption~\ref{NA} there is no asymptotic arbitrage of the first kind for transaction costs $\la_n$ if and only if
there exists a sequence $({Q}^n)$, with ${Q}^n\in \mathcal{M}^n(\la_n)$ for each $n$, such that $\pn\con({Q}^n)$.
\end{theorem}

\begin{theorem}\label{2}
Under Assumption~\ref{NA} there is no asymptotic arbitrage of the second kind for transaction costs $\la_n$ if and only if for each $\ep>0$ there is $\delta>0$ and a sequence of measures $({Q}^n)$, with ${Q}^n\in \mathcal{M}^n(\la_n)$ for each $n$,
such that for each $A^n\in\mathcal{F}^n$ with $P^n(A^n)<\delta$ we have that ${Q}^n(A^n)<\ep$, for all $n$.
\end{theorem}

The property of the sequence of sets $\left(\mathcal{M}^n(\la_n)\right)$ in the above theorem appeared already in
\cite{K:S:1996} in an analogous result on NAA2 without transaction costs and in \cite{Kab:Kra:1998} where it was named {\it weak contiguity}.
It is clear that with transaction costs the condition cannot look nicer as the zero transaction costs case is
included.

\begin{theorem}\label{strong}
Under Assumption~\ref{NA} there exists a strong asymptotic arbitrage for transaction costs $\la^n$ if and only if there exists a subsequence $n_k$ and sequence of sets $A^k\in\mathcal{F}^{n_k}$ such that $P^{n_k}(A^{k})\to1$
and $\lim_{k\to\infty}\sup_{Q^{n_k}\in \mathcal{M}(\la_{n_k})}Q^{n_k}(A^k)=0$.
 This means that  $\pn$ is entirely asymptotically separable from the sequence of the upper envelope of the sets $\mathcal{M}(\la_{n_k})$.
\end{theorem}

\section{The Proofs of Theorems~\ref{1}, \ref{2} and \ref{strong} }\label{proofs}
In order to prove the theorems we will first give
characterizations of NAA1($\la_n)$ and NAA2($\la_n$) (Lemma~\ref{L} below). These characterizations will then
immediately give Theorem~\ref{2} by the use of Propostion~\ref{HS2}. For the proof of Theorem~\ref{1} we will
need Proposition~\ref{HS1} and, in addition, that the set $\mathcal{M}(\la_n)$ is closed under countable convex
combinations. As this result is in fact a monotone convergence result for local martingales we will prove it
separately in the next section, see Section~\ref{locmart}.

\begin{lemma}\label{L}
\begin{enumerate}
\item[]
\item[(1)]
There does not exists an asymptotic arbitrage of the first kind with transaction costs $\la_n$ if and only if for
each $\ep>0$ there exists $\delta>0$ such that for each $n$ and $A^n\in\mathcal{F}^n$ with $P^n(A^n)\geq\ep$
there is $\widetilde{Q}^n\in\mathcal{M}^n(\la_n)$ with $\widetilde{Q}^n(A^n)\geq\delta$.
\item[(2)]
There does not exists an asymptotic arbitrage of the second kind with transaction costs $\la_n$ if and only if for
each $\ep>0$ there exists $\delta>0$ such that for each $n$ and $A^n\in\mathcal{F}^n$ with $P^n(A^n)<\delta$
there is $\widetilde{Q}^n\in\mathcal{M}^n(\la_n)$ with $\widetilde{Q}^n(A^n)<\ep$.
\end{enumerate}
\end{lemma}

\begin{proof}[Proof of Lemma~\ref{L}]
(1) Assume that NAA1($\la_n)$ holds. We will show the $\ep$--$\delta$--condition. Assume to the contrary that there exists $\ep>0$ such that for all $\delta>0$  there is $n$ and $A^n\in\mathcal{F}^n$ with $P^n(A^n)\geq \ep$ but $\sup_{Q\in\mathcal{M}^n(\la_n)}Q(A^n)<\delta$. Choose $\delta_k<\frac12$ such that $\delta_k\to0$ and assume that the corresponding $n_k\to\infty$ (this is w.l.o.g., see the remark below). Define
$$f^k=\frac1{\sqrt{\delta_k}}\ind_{A^{n_k}}-2\sqrt{\delta_k}\ind_{\Omega^{n_k}\setminus A^{n_k}}.$$
For each $\widetilde{Q}^{n_k}\in\mathcal{M}^{n_k}(\la_{n_k})$ and the corresponding $\widetilde{S}^{n_k}$ we have that
\begin{align}
\E_{\widetilde{Q}^{n_k}}[f^k]
&=\frac1{\sqrt{\delta_k}}\widetilde{Q}^{n_k}(A^{n_k})-2\sqrt{\delta_k}(1-\widetilde{Q}^{n_k}(A^{n_k}))\nonumber\\
&\leq \frac1{\sqrt{\delta_k}}\delta_k-2\sqrt{\delta_k}(1-\delta_k)
=-\delta_k^{\frac12}+2\delta_k^{\frac32}\leq 0,\nonumber
\end{align}
which holds as $\delta_k<\frac12$.
By the Superreplication Theorem (see Theorem~\ref{super} of the Appendix) $f^k=V_{T^{n_k}}^{\la_{n_k},S^{n_k}}(\phi^k)$ for some $\la_{n_k}$--admissible strategy $\phi^k$. Note that $f^k\geq-2\sqrt{\delta_k}$ a.s., so we get by Proposition~\ref{schach} of the Appendix that $V_t^{\la_{n_k},S^{n_k}}(\phi^k)\geq-2\sqrt{\delta_k}$ a.s. for all $0\leq t\leq T^{n_k}$.
Moreover, we have that
$$P^{n_k}(V_{T^{n_k}}^{\la_{n_k},S^{n_k}}(\phi^k)\geq \frac1{\sqrt{\delta_k}})=P^{n_k}(A^{n_k})\geq\ep.$$ This gives an AA1($\la_{n}$), which is a contradiction.
\newline\newline
\noindent In order to prove the other direction, assume that the $\ep$--$\delta$--condition of (1) holds. Assume that there is an AA1($\la_n$). Then there exists $\ep>0$ and a subsequence again denoted by $n$ and admissible integrands $\phi^n$ such that $P^n(V_{T_n}^{\la_n,S^n}(\phi^n)\geq C_n)\geq\ep$, for all $n$, and $V^{\la_n,S^n}_t(\phi^n)\geq -c_n$ a.s. for $0\leq t\leq T_n$. By assumption there exists $\delta>0$ such that, for each $n$, there is $\widetilde{Q}^n\in \mathcal{M}^n(\la_n)$ with
$$\widetilde{Q}^n(V^{\la_n,S^n}_{T_n}(\phi^n)\geq C_n)\geq\delta.$$
Let $\widetilde{S}^n$ be the corresponding $\widetilde{Q}^n$--local martingale, then, by Lemma~\ref{blum} of the Appendix we have that
$$V^{0,\widetilde{S}^n}_t(\phi^n)\geq  V^{\la_n, S^n}_t(\phi^n)\text{ a.s.}$$ for each $0\leq t\leq T_n$. Therefore the strategies $\phi^n$ define an asymptotic arbitrage with transaction costs 0 for the assets $\widetilde{S}^n$. By Lemma~\ref{grr} we have that
$$\E_{\widetilde{Q}^n}[V^{0,\widetilde{S}^n}_{T_n}(\phi^n)]\leq0.$$
On the other hand, we get
$$\E_{\widetilde{Q}^n}[V^{0,\widetilde{S}^n}_{T_n}(\phi^n)]\geq C_n\widetilde{Q}^n(V^{\la_n,S^n}_{T_n}(\phi^n)\geq C_n)-c_n\geq C_n\delta-c_n>0,$$
for $n$ large enough. This is a contradiction.
\newline\newline
\noindent (2) Assume that NAA2($\la_n$) holds. We will show the $\ep$--$\delta$--condition.
Assume to the contrary that there exists $\ep>0$ such that for all $\delta>0$  there is $n$ and $A^n\in\mathcal{F}^n$ with $P^n(A^n)<\delta$ but $\inf_{Q\in\mathcal{M}^n(\la_n)}Q(A^n)\geq\ep$. Choose  $\delta_k\to0$ and assume that the corresponding $n_k\to\infty$ (w.l.o.g. we can choose such $\delta_k$, see again the remark below). Let
$$f^k=\ep\ind_{\Omega^{n_k}\setminus A^{n_k}}-\ind_{A^{n_k}}.$$
Similarly as in the proof of part (1), for each $\widetilde{Q}^{n_k}\in\mathcal{M}^{n_k}(\la_{n_k})$ and the corresponding $\widetilde{S}^{n_k}$ we have that
\begin{align}
\E_{\widetilde{Q}^{n_k}}[f^k]
&=\ep(1-\widetilde{Q}^{n_k}(A^{n_k}))-\widetilde{Q}^{n_k}(A^{n_k})\nonumber\\
&\leq \ep(1-\ep)-\ep=-\ep^2\leq0.\nonumber
\end{align}
Again by the Superreplication Theorem~\ref{super} $f^k= V^{\la_{n_k}, S^{n_k}}_{T_{n_k}}(\phi^k)$ for some admissible strategy $\phi^k$. Note that $f^k\geq-1$ a.s., so we get by Proposition~\ref{schach} of the Appendix that $V^{\la_{n_k},S^{n_k}}_t(\phi^k)\geq-1$ a.s. for all $0\leq t\leq T^{n_k}$.
Moreover, we have that
$$P^{n_k}(V^{\la_{n_k},S^{n_k}}_{T_{n_k}}(\phi^k)\geq \ep)=1-P^{n_k}(A^{n_k})\geq1-\delta_k.$$ This gives an AA2($\la_{n}$).
\newline\newline
\noindent In order to prove the other direction, assume that the $\ep$--$\delta$--condition of (2) holds. Assume that there is an AA2($\la_n$). Then there exists $\alpha>0$ a subsequence again denoted by $n$ and admissible integrands $\phi^n$ such that
$V^{\la_n, S^{n}}_t(\phi^n)\geq -1$ a.s. for $0\leq t\leq T_n$ and
$P^n(V^{\la_n,S^n}_{T_n}(\phi^n)\geq \alpha)=:1-\delta_n$, for all $n$, where $\delta_n\to0$. By assumption, for $\ep_k=2^{-k}$ there is $n(k)$ (such that $\delta_{n(k)}$ small enough) and $\widetilde{Q}^{k}\in \mathcal{M}^{n_k}(\la_{n_k})$ with
$$\widetilde{Q}^{k}(V^{\la_{n_k}, {S^{n_k}}}_{T_{n_k}}(\phi^{n_k})< \alpha)<2^{-k}.$$
Let $\widetilde{S}^k$ be the corresponding $\widetilde{Q}^k$--local martingale, then, again by Lemma~\ref{blum}
of the Appendix we have that
$$V^{0,\widetilde{S}^k}_t(\phi^{n_k})\geq  V^{\la_{n_k}, S^{n_k}}_t(\phi^{n_k},S^{n_k})\text{ a.s.}$$ for each $0\leq t\leq T_k$. Therefore the strategies $\phi^{n_k}$ define an asymptotic arbitrage without transaction costs 0 for the assets $\widetilde{S}^k$. Again by Lemma~\ref{grr} we have that
$$\E_{\widetilde{Q}^{k}}[V^{0,\widetilde{S}^k}_{T_{n_k}}(\phi^{n_k})]\leq0.$$
On the other hand, we get
\begin{align}
\E_{\widetilde{Q}^k}[V^{0,\widetilde{S}^k}_{T_{n_k}}(\phi^{n_k})] &\geq -1\widetilde{Q}^k(V^{\la_{n_k},S^{n_k}}_{T_{n_k}}(\phi^{n_k})<\alpha)+\alpha \widetilde{Q}^n(V^{\la_{n_k},S^{n_k}}_{T_{n_k}}(\phi^{n_k})\geq\alpha)\nonumber\\
&\geq -2^{-k}+\alpha(1-2^{-k})>0,\nonumber
\end{align}
for $k$ large enough. This is a contradiction.
\end{proof}

\begin{remark}
\begin{enumerate}
\item In the first part of the proof of (1) we can choose $\delta_k\to 0$ such that the corresponding $n_k\to\infty$. Indeed, choose $\delta_k<2^{-k}$ and assume that $n_k\leq C$, for $C<\infty$. Then there is $N$ such that there is another subsequence (again denoted by $n_k$ with $n_k=N$ for all $k$).  Let $P=P^N$. This means that, for each $k$ there is $A_k\in\mathcal{F}^N$ with $P(A_k)\geq \ep$
    and $Q(A_k)\leq\delta_k$ for all $Q\in\mathcal{M}^N(\la_n)$. Define $A=\limsup A_k=\bigcap_{m}\bigcup_{k\geq m}A_k$. Then $P(A)\geq\limsup P(A^k)\geq\ep$.
    On the other hand, we have that, for each $Q$ and for each $m$,
    $$Q(A)\leq Q(\bigcup_{k\geq m}A_k)\leq \sum_{k=m}^{\infty}Q(A_k)<\sum_{k=m}^{\infty}\delta_k<2^{-(m-1)}.$$
    As this holds for any $m$, we have that $Q(A)=0$ for each $Q\in\mathcal{M}^N(\la_n)$. Hence $P\not\ll Q$ for any $Q$. This is a contradiction to Assumption~\ref{NA}, as there has to be at least one equivalent measure.
 \item In the first part of the proof of (2) we can choose $\delta_k\to 0$ such that the corresponding $n_k\to\infty$. Indeed, choose again $\delta_k<2^{-k}$ and assume that $n_k\leq C$, for $C<\infty$. Then there is $N$ such that there is another subsequence (again denoted by $n_k$ with $n_k=N$ for all $k$). Let $P=P^N$ This means that, for each $k$ there is $A_k\in\mathcal{F}^N$ with $P(A_k)<\delta_k$
    and $Q(A_k)\geq\ep$ for all $Q\in\mathcal{M}^N(\la_n)$. Define again $A=\limsup A_k=\bigcap_{m}\bigcup_{k\geq m}A_k$. Then
    $P(A)\leq \sum_{k=m}^{\infty}P(A^k)\leq \sum_{k=m}^{\infty}\delta_k<2^{-(m-1)}$. As this holds for any $m$ we have that $P(A)=0$.
    On the other hand, we have that, for each $Q$, $\E_Q[A]\geq \limsup Q(A_k)\geq \ep$. Hence $Q\not\ll P$ for each $Q\in\mathcal{M}^N(\la_n)$.
     This is again a contradiction to Assumption~\ref{NA}, as there has to be at least one equivalent measure.
    \end{enumerate}
\end{remark}

\begin{proof}[Proof of Theorem~\ref{1} and Theorem~\ref{2}:]
The proof of Theorem~\ref{2} follows by applying Lemma~\ref{L}, (2), to the sets $\mathcal{M}^n(\la_n)$ and then applying Proposition~\ref{HS2} of the Appendix.
\newline\newline
\noindent
For the proof of Theorem~\ref{1} note that Lemma~\ref{L}, (1),  and Proposition~\ref{HS1} give for each $\ep>0$ a $\delta>0$ and a sequence $\widetilde{Q}^{n,\ep}\in
\mathcal{M}^n(\la_n)$ such that for all $n$ and $A^n\in\mathcal{F}^n$ with $\widetilde{Q}^{n,\ep}(A^n)<\delta$ we have that $P^n(A^n)<\ep$. In order to get rid of the dependence on $\ep$, we argue as in \cite{K:S:1996} (for example). Take $\ep_m=\frac1{m}$, $m\in\N$, and define, for each $n$,
$$\widetilde{Q}^n=\sum_{m=1}^{\infty}2^{-m}\widetilde{Q}^{n,\ep_m}.$$
Then $\widetilde{Q}^n\in\mathcal{M}^n(\la_n)$ as $\mathcal{M}^n(\la_n)$ is closed under countable convex combinations, see Corollary~\ref{Conv}. It is not hard to see that, moreover, $(P^n)\con (\widetilde{Q}^n)$, see \cite{K:S:1996}, proof of Theorem~2.1.
\end{proof}

Finally we present the  easy proof of Theorem~\ref{strong}.

\begin{proof}[Proof of Theorem~\ref{strong}]

Suppose that there is a strong asymptotic arbitrage opportunity $V^{\la{n_k},S^{n_k}}(\phi^k)$ with transaction costs $\la_{n_k}$.
Assume that for each sequence $A^n\in\mathcal{F}^n$ with $P^n(A^n)\to1$
there exists a sequence $\widetilde{Q}^n\in \mathcal{M}^n(\la_n)$ such that  we have that $\widetilde{Q}^n(A^n)\geq \alpha$ for some $\alpha>0$. Define $A_k=\{V^{\la_{n_k}, S^{n_k}}_{T_{n_k}}(\phi^k)\geq C_k\}$. Let $(\widetilde{Q}^{n_k}, \widetilde{S}^{n_k})$ be the above corresponding sequence of $\la_{n_k}$--CPS for the sequence $(A_k)$. As $P^{n_k}(A^k)\to1$ by assumption, we have that $\widetilde{Q}^{n_k}(A^k)\geq\alpha$, for all $k$, and some $\alpha>0$. As above we have that
$$V_t^{0,\widetilde{S}^{n_k}}(\phi^k)\geq V_t^{\la_{n_k},S^{n_k}}(\phi^k)\text{ a.s. for all $t$},$$
and
$$\E_{\widetilde{Q}^{n_k}}[V^{0,\widetilde{S}^{n_k}}_{T_{n_k}}(\phi^k)]\leq 0.$$
On the other hand we get
$$\E_{\widetilde{Q}^{n_k}}[V^{0,\widetilde{S}^{n_k}}_{T_{n_k}}(\phi^k)]\geq C_k\widetilde{Q}^{n_k}(A^k)-c_k\geq C_k\alpha-c_k>0,$$
for $k$ large enough. This is a contradiction.
\newline\newline
For the other direction assume that $\pn$  is entirely asymptotically separable from
the sequence of the upper envelopes of the sets $\mathcal{M}^n(\la_n)$. Let $n_k$ be a subsequence and $A^k$ be the sequence of sets $A^k\in\mathcal{F}^{n_k}$ with $P^{n_k}(A^k)\to1$. By assumption $\sup_{Q^{n_k}\in\mathcal{M}^{n_k}(\la_{n_k})}\Q^{n_k}(A^k)=:\delta_k\to0$.   W.l.o.g. $\delta_k<\frac12$, for all $k$. Define
$$f_k=\frac1{\sqrt{\delta_k}}\ind_{A^k}-2\sqrt{\delta_k}.$$
We have that
\begin{align}
\E_{\widetilde{Q}^{n_k}}[f^k]
&=\frac1{\sqrt{\delta_k}}\widetilde{Q}^{n_k}(A^{k})-2\sqrt{\delta_k}(1-\widetilde{Q}^{n_k}(A^{k}))\nonumber\\
&\leq \frac1{\sqrt{\delta_k}}\delta_k-2\sqrt{\delta_k}(1-\delta_k)
=-\delta_k^{\frac12}+2\delta_k^{\frac32}\leq 0.\nonumber
\end{align}

Note that $f^k\ge -2\sqrt{\delta_k}$ and $P^{n_k}(f^k\geq \frac1{\sqrt{\delta_k}})=P^{n_k}(A^k)\to1.$ As in the proof of Lemma~\ref{L}, part (1), this gives a strong asymptotic arbitrage.

\end{proof}

\section{Monotone convergence for local martingales and an application to the set $\mathcal{M}(\la)$}\label{locmart}
We start with an observation on local martingales which will provide a monotone convergence result for local martingales. Note that in this section we do not assume that the local martingales have continuous paths, we only assume that we take the c\`adl\`ag modifications.
\begin{lemma}\label{red}
Let $(N_t)_{0\leq t\leq T}$ be a local martingale and $(U_t)_{0\leq t\leq T}$ be a martingale such that, for all $0\leq t\leq T$,
$$0\leq N_t\leq U_t\quad\text{a.s.}$$
Then $(N_t)_{0\leq t\leq T}$ is a martingale.
\end{lemma}

\begin{proof}
As $(U_t)_{0\leq t\leq T}$ is defined on the closed time interval $[0,T]$ it is a uniformly integrable martingale. Moreover, the family $\{U_{\tau}\}$ where $\tau$ runs through the set of all stopping times, which are $\leq T$ a.s., is  uniformly integrable as well. Let $\tau_k\uparrow\infty$ be a sequence of stopping times such that, for each $k$, the stopped process $N^{\tau_k}$ is a martingale. Then, we have that, for each $t\leq T$ and for each arbitrary constant $C>0$,
\begin{equation}\label{ui}
\sup_{k}E[N_{t\wedge\tau_k}\ind_{\{N_{t\wedge\tau_k}\geq C\}}]\leq
\sup_{k}E[U_{t\wedge\tau_k}\ind_{\{U_{t\wedge\tau_k}\geq C\}}].
\end{equation}
This holds as $N_t\leq U_t$ a.s. for each $t$ and because of the right--continuity of both processes, i.e. let $A_q=\{N_q>U_q\}$, for $q\in([0,T]\cap\Q)\cup\{T\}$. By assumption $P(A_q)=0$ for each $q$. Define $A=\bigcup_{q\in[0,T]\cap\Q)\cup\{T\}}A_q$, then still $P(A)=0$.  Let $\tilde{\Omega}=(\Omega\setminus A)\cap \Omega_N\cap\Omega_U$, where $\Omega_N$, $\Omega_U$ are the sets of probability 1 where the paths of $N$, $U$, respectively are right--continuous, then $P(\tilde{\Omega})=1$. Let $t$ be arbitrary and $\omega\in\tilde{\Omega}$. Then
$$N_t(\omega)=\lim_{q_n\downarrow t} N_{q_n}(\omega)\leq \lim_{q_n\downarrow t }U_{q_n}(\omega)= U_{t}(\omega).$$
So (\ref{ui}) indeed holds as $P(N_t\leq U_t\text{ for all $t$})\geq P(\tilde{\Omega})=1$. As $\{U_{t\wedge\tau_k}, k\in\N\}$ is uniformly integrable, the right hand side of (\ref{ui}) tends to 0 for $C\to\infty$, therefore $\{N_{t\wedge\tau_k}, k\in\N\}$ is uniformly integrable as well. Moreover, as $\tau_k\uparrow\infty$, we have that $N_{t\wedge\tau_k}\to N_t$ a.s. and hence, by uniform integrability, the convergence holds in $L^1$ as well. Therefore we get, for $s<t$,
\begin{align}
E[N_t|\mathcal{F}_s] &=E[\lim_{k\to\infty}N_{t\wedge\tau_k}|\mathcal{F}_s]\nonumber\\
&=
\lim_{k\to\infty}E[N_{t\wedge\tau_k}|\mathcal{F}_s]\nonumber\\
&=\lim_{k\to\infty}N_{s\wedge\tau_k}=N_s,\nonumber
\end{align}
by the above and by the martingale property of $N^{\tau_k}$. Hence $N$ is a martingale.
\end{proof}

Moreover we will use the following well-known Lemma. For the convenience of the reader we give the easy proof.
\begin{lemma}\label{rightcon}
Let $(X_t)$ be a c\`adl\`ag supermartingale with $X_t\geq 0$ a.s. for each $t$. Then $t\mapsto E[X_t]$ is right--continuous.
\end{lemma}

\begin{proof}
Let $t_n\downarrow t$. By right continuity of the paths we have that $X_{t_n}\to X_t$ a.s. By the supermartingale property we have, for each $n$, as $t_n\geq t$,
$$0\leq E[X_{t_n}|\mathcal{F}_t]\leq X_t.$$
So we get that
\begin{align}
\lim_{n\to\infty} E[X_{t_n}] &=\lim_{n\to\infty}E[E[X_{t_n}|\mathcal{F}_t]]\nonumber\\
&=E[\lim_{n\to\infty}E[X_{t_n}|\mathcal{F}_t]]\label{y1}\\
&\geq E[E[\lim_{n\to\infty}X_{t_n}|\mathcal{F}_t]]=E[X_t]\label{y2},
\end{align}
where (\ref{y1}) follows by dominated convergence and (\ref{y2}) by Fatou and by the right--continuity of the paths.
On the other hand, as $t_n\geq t$, for all $n$, we have that
$$E[X_{t_n}]\leq E[X_t],$$
by the supermartingale property. Therefore $\lim_{n\to\infty}E[X_{t_n}]=E[X_t]$, hence $t\mapsto E[X_t]$ is right--continuous.
\end{proof}

Now we have all tools to prove a monotone convergence result for local martingales. From this result the closedness for countable convex combinations of the set $\mathcal{M}(\la)$ immediately follows.

\begin{proposition}[Monotone convergence of local martingales]\label{mon}
Let $(M^n_t)_{0\leq t\leq T}$ be a sequence of local martingales on $(\Om,\mathcal{F},(\mathcal{F}_t)_{0\leq t\leq T},P)$ such that $M^{n+1}_t\geq M^n_t\geq0$, for all $t$ and $n$, and such that $M^n_t\to M_t$ a.s., for all $t$, and $E[M_0]=a<\infty$. Then $(M_t)_{0\leq t\leq T}$ is a local martingale.
\end{proposition}

In the formulation of the corollary we omit the superscript $n$, but of course this holds, for each $n$.
\begin{corollary}\label{Conv}
 $\mathcal{M}(\la)$ is closed under countable convex combinations.
\end{corollary}

\begin{proof}[Proof of Proposition~\ref{mon}]
As $(M^n_t)$ is a non-negative local martingale it is a supermartingale, for each $n$.
We claim that the limiting process is a c\`adl\`ag supermartingale. Indeed, by monotone convergence and by the supermartingale property of each $M^n$ we get, for $s<t$,
$$E[M_t|\mathcal{F}_s] =E[\lim_{n\to\infty}M^n_t|\mathcal{F}_s]
=\lim_{n\to\infty}E[M^n_t|\mathcal{F}_s]\nonumber\leq \lim_{n\to\infty}M^n_s=M_s,$$
in particular $E[M_t]\leq E[M_0]=a<\infty$.
So the supermartingale property is clear.

Moreover, we will show that $t\mapsto E_P[M_t]$ is right--continuous (which implies that there is a c\`adl\`ag modification of $M$ so that we can assume that $M$ has c\`adl\`ag paths, see \cite{R:Y}). Indeed, let $q_k\downarrow t$. First
 define $a_{n,k}:=E[M^n_{q_k}]$ which is a double sequence in $\R$, where $0\leq a_{n,k}\leq E[M_0]=a$. Define the
order $(n,k)\preceq (m,l)$ if $n\leq m$ and $k\leq l$. Then we have that for $(n,k)\preceq (m,l)$
$$a_{n,k}=E[M^n_{q_k}]\leq E[M^n_{q_l}]\leq E[M^m_{q_l}]=a_{m,l},$$
where the first inequality is the supermartingale property of $M^n$ (as $q_k\geq q_l$ for $k\leq l$). The second inequality is as by assumption $M^n_{t}\leq M^m_{t}$ for all $t$ and $n\leq m$. Therefore $(a_{n,k})_{(n,k)}$ is a bounded monotone  net, hence it has a unique limit which is
\begin{equation}\lim_{k\to\infty}\lim_{n\to\infty}a_{n,k}=\lim_{n\to\infty}\lim_{k\to\infty}a_{n,k}=\sup_{n,k}a_{n,k}.\label{interchange}\end{equation}
Therefore we have that
\begin{align}\lim_{k\to\infty}E[M_{q_k}] &=\lim_{k\to\infty}E[\lim_{n\to\infty}M^n_{q_k}]\nonumber\\
&=\lim_{k\to\infty}\lim_{n\to\infty}E[M^n_{q_k}]\label{x1}\\
&=\lim_{n\to\infty}\lim_{k\to\infty}E[M^n_{q_k}]\label{x2}\\
&=\lim_{n\to\infty}E[M^n_{t}]=E[M_t],\label{x3}
\end{align}

where (\ref{x1}) is monotone convergence and the interchange of limits in (\ref{x2}) follows by (\ref{interchange}).
The first equality of (\ref{x3}) follows by Lemma~\ref{rightcon} applied to each supermartingale $M^n$ and the second equality of \ref{x3} is again monotone convergence.

So we already know that $M$ is a c\`adl\`ag supermartingale. We will show now that is even a local martingale. The Doob--Meyer decomposition gives that
$$M_t=U_t-A_t,$$
where $(U_t)_{0\leq t\leq T}$ is a local martingale under $P$ with $U_0=M_0$ and $(A_t)_{0\leq t\leq T}$ is an increasing process with $A_0=0$. As, for each $n$, we have that $M^n_t\leq M_t$ a.s., we get, in particular that
$$0\leq M^n_t\leq U_t\quad\text{a.s., $0\leq t\leq T$.}$$
Let $\tau_k\uparrow\infty$ be a sequence of stopping times such that the stopped process $U^{\tau_k}$ is a martingale, for each $k$. By Lemma~\ref{red} we have that $(\tau_k)$ is a reducing sequence for $M^n$, each $n$, as well. That means that $(M^n)^{\tau_k}$ is a martingale, for each $n$. So, for each $n$ and for each $k$, we have that
$$E[M^n_0]=E[M^n_{t\wedge\tau_k}]\leq E[U_{t\wedge\tau_k}]-E[A_{t\wedge\tau_k}]=E[M_0]-E[A_{t\wedge\tau_k}].$$
So, for each $n$,
$$0\leq E[A_{t\wedge\tau_k}]\leq E[M_0-M_0^n].$$
As $E[M_0-M_0^n]\to0$ for $n\to\infty$

\begin{equation}E[A_{t\wedge\tau_k}]=0,\label{x4}\end{equation}
a.s. for each $k$. Moreover, as $A$ is increasing, we have that $A_{t\wedge\tau_k}\uparrow A_t$ a.s., so (\ref{x4}) gives that $E[A_t]=0$, for each $t$. But as $A_t\geq0$ we therefore get $A_t=0$ a.s. for each $t$. So we have that $M_t=U_t$ a.s., and $U$ was a local martingale by assumption. This finishes the proof.
\end{proof}

\begin{proof}[Proof of Corollary~\ref{Conv}]
Let ${Q}^m$, $m=1,2,\dots$, be a sequence of measures in  $\mathcal{M}(\la)$ and let $(\alpha_m)$ be a sequence of   real numbers $\alpha_m\geq0$, for all $m$, and $\sum_{m=1}^{\infty}\alpha_m=1$. Define ${Q}=\sum_{m=1}^{\infty}\alpha_m{Q}^m$. We have to show that $Q\in\mathcal{M}(\la)$.

It is clear that $Q\sim P$, as for $A$ with $P(A)>0$ we have that $Q^m(A)>0$, for all $m$ and hence $Q(A)>0$. If $Q(A)>0$ there exists $m$ such that $Q^m(A)>0$ and hence $P(A)>0$.
\newline\newline
\noindent
Define $Z^m_t=E_P[\frac{dQ^m}{dP}|\mathcal{F}_t]$ and $Z_t=E_P[\frac{dQ}{dP}|\mathcal{F}_t]$, for $0\leq t\leq T$. By monotone convergence we get that
\begin{align}
\frac{dQ}{dP}&=\sum_{m=1}^{\infty}\alpha_m\frac{dQ^m}{dP}\text{ and}\label{xx}\\
Z_t &=\sum_{m=1}^{\infty}\alpha_mZ^m_t.\label{xxx}
\end{align}
 By assumption, for each $m$, there exists a $Q^m$--local martingale $\widetilde{S}^m$ such that $(Q^m,\widetilde{S}^m)$ is a $\la$--CPS.
We will show that there is $\widetilde{S}$ such that $({Q},\widetilde{S})$ is a $\la$--CPS as well. Note that as  $\widetilde{S}^m$
is a  $Q^m$--local martingale, we have that $(\widetilde{S}^m_tZ^m_t)_{0\leq t\leq T}$ is a local martingale under $P$, for each $m$.
Define
$$ \widetilde{S}_t:=\frac{\sum_{m=1}^{\infty}\alpha_m\widetilde{S}^m_tZ^m_t}{Z_t}.$$
In order to show that $\widetilde{S}$ is a local martingale with respect to $Q$ we will show that $(M_t)_{0\leq t\leq T}$ where
$$M_t:=\widetilde{S}_tZ_t=\sum_{m=1}^{\infty}\alpha_m\widetilde{S}^m_tZ^m_t$$
is a local martingale with respect to $P$. Define
$$M^N_t=\sum_{m=1}^{N}\alpha_m\widetilde{S}^m_tZ^m_t,$$
then it is clear that $M^N_t\uparrow M_t$ a.s., for each $0\leq t\leq T$. Moreover, as $M^N$ is a finite sum of c\`adl\`ag  $P$--local martingales, it is a c\`adl\`ag  $P$--local martingale as well (this holds for each $N$).
Moreover, by monotone convergence and as each $\widetilde{S}^m$ lies in the bid ask spread of $S$ we get
$$E[M_0]=\lim_{N\to\infty}E[M^N_0]=\sum_{m=1}^{\infty}\alpha_mE[\widetilde{S}^m_0Z^m_0]\leq
\sum_{m=1}^{\infty}\alpha_m{S}_0=S_0<\infty.$$
Therefore, by Proposition~\ref{mon}, $(M_t)_{0\leq t\leq T}$ is a $P$--local martingale.
\end{proof}

\section{An Example}\label{ex}
We present an example which allows a
strong asymptotic arbitrage if there are no transaction costs. If we impose the same small transaction costs $\la>0$ on each market $n$ then there does not exist any form of asymptotic arbitrage ($\la>0$ can be arbitrarily small). If we impose transaction costs $\la_n$ on market $n$ such that $\la_n\to0$ we present a lower bound for $\la_n$ such that still there does not exist any form of asymptotic arbitrage.

The idea of the example is inspired by the example $S_t=\exp{(W_t-\frac{t}{2}+t^{1/2})}$, which allows immediate arbitrage and is an easy variant of an example in \cite{De:Sch}.

Let $(W_t)_{t\geq0}$ be a standard Brownian motion on $(\Omega,\mathcal{F},(\mathcal{F}_t)_{t\geq 0},P)$.
Let $T_n$ be any increasing sequence of positive real numbers with $\lim_{n\to\infty}T_n=\infty$.
For each $n\geq1$ we define the following financial asset $(S^n_t)_{0\leq t\leq T_n}$.

\begin{equation}
\frac{dS^n_t}{S^n_t}=dW_t+\frac1{T_n(T_n-t+\alpha_n)^{\frac12}}dt,\label{S^{Tn}}
\end{equation}

where $\alpha_n=e^{-T_n^{2+\ep}}$ for some fixed $\ep>0$.  Moreover we choose $S^n_0=1$ for all $n$.
The market $n$ is defined by $(S^n_t)_{0\leq t\leq T_n}$ on $(\Omega,\mathcal{F}_{T_n},(\mathcal{F}_t)_{0\leq t\leq T_n},P|_{\mathcal{F}_{T_n}})$
and a riskless asset which is identically 1.
For each $n$, there exists a unique equivalent  martingale measure for $Q^n$ for $S^n$
given by $\frac{dQ^n}{dP}=Z^n_{T_n}$ where

\begin{equation}
Z^n_{T_n}=\exp\left(-\int_0^{T_n} \frac1{T_n(T_n-t+\alpha_n)^{\frac12}}dW_t-\frac1{2T_n^2}\int_{0}^{T_n}\frac1{T_n-t+\alpha_n}dt
\right)
\end{equation}

Indeed, for each $n$, Novikov's condition holds as

\begin{equation}\frac1{T_n^2}\int_{0}^{T_n}\frac1{T_n-t+\alpha_n}dt=\frac1{T_n^2}\left(\ln(T_n+\alpha_n)-\ln(\alpha_n)\right)=
\frac{\ln(T_n+\alpha_n)}{T_n^2}+T_n^{\ep}
<\infty.\label{Novikov_S}\end{equation}

As there exists an equivalent martingale measure there does not exist any form of arbitrage on the market $n$ (with 0 transaction costs and therefore with strictly positive transaction costs as well).

However, as  in (\ref{Novikov_S}) we have that $\lim_{n\to\infty}\left(\frac{\ln(T_n+\alpha_n)}{T_n^2}+T_n^{\ep}\right)=+\infty$, we expect some kind of asymptotic arbitrage behaviour for $n\to\infty$.
The following result and proof is analogous to Theorem 4.5 from \cite{F:S}.

\begin{proposition}
There exists a strong asymptotic arbitrage.
\end{proposition}

\begin{proof}
It is enough to show that there exist sets $A_n\in\mathcal{F}_{T_n}$ with $\lim_{n\to\infty}P(A_n)=1$ and $\lim_{n\to\infty}Q^n(A_n)=0$. Then the sequences of measures are entirely asymptotically separable which is equivalent to the existence of a strong asymptotic arbitrage, see \cite{Kab:Kra:1998}.

Let $\gamma<\frac12$ and $A_n=\{Z^n_{T_n}<(\tfrac{T_n+\alpha_n}{\alpha_n})^{-\frac{\gamma}{T_n^2}}\}$. Then we have that

\begin{align}
P(A_n) &=P\left(-\int_0^{T_n} \frac1{T_n(T_n-t+\alpha_n)^{\frac12}}dW_t-\frac1{2T_n^2}\ln\left(\frac{T_n+\alpha_n}{\alpha_n}\right)<
-\frac{\gamma}{T_n^2}\ln\left(\frac{T_n+\alpha_n}{\alpha_n}\right)\right)\nonumber\\
&=P\left(\frac{-\int_0^{T_n} \frac1{T_n(T_n-t+\alpha_n)^{\frac12}}dW_t}{\frac1{T_n}\sqrt{\ln\left(\frac{T_n+\alpha_n}{\alpha_n}\right)}}<\left(\frac12-\gamma\right)\frac1{T_n}\sqrt{\ln\left(
\frac{T_n+\alpha_n}{\alpha_n}\right)}
\right)\nonumber\\
&=\Phi\left(\left(\frac12-\gamma\right)\frac1{T_n}\sqrt{\ln\left(\frac{T_n+\alpha_n}{\alpha_n}\right)}\right)\nonumber\\
&=\Phi\left(\left(\frac12-\gamma\right)\sqrt{\frac{\ln(T_n+\alpha_n)}{T_n^2}+T_n^{\ep}}\right),\nonumber
\end{align}
where $\Phi$ denotes the distribution function of standard normal distribution. As $\gamma<\frac12$ and $T_n^{{\ep}
}\to\infty$ we have that $\lim_{n\to\infty}P(A_n)=1$.

%

Note that

$$\left(\frac{T_n+\alpha_n}{\alpha_n}\right)^{-\frac{\gamma}{T_n^2}}=
\exp\left(-\gamma\left(\frac{\ln(T_n+\alpha_n)}{T_n^2}+T_n^{\ep}\right)\right)\to0,$$

for $n\to\infty$. We have that

$$Q^n(A_n)=E_{P}[Z^n_{T_n}\ind_{A_n}]<\left(\frac{T_n+\alpha_n}{\alpha_n}\right)^{-\frac{\gamma}{T_n^2}}P(A_n),$$

and therefore $\lim_{n\to\infty}Q^n(A_n)=0$.
\end{proof}

We will now see that if we have transaction costs $\lambda_n>0$ on market $n$ such that $\la_n\to0$ not too fast,  then the asymptotic arbitrage disappears.

\begin{theorem}\label{Ex1} Fix transaction costs $\la_n$ on market $n$ where $\la_n>2(1-e^{-\gamma_n})$ and $\gamma_n\to0$ is defined as
$$\gamma_n=\frac{2}{T_n}e^{-\frac{T^2_n}{2}}\big(1+e^{-T_n^2(T_n^{\ep}-1)}\big)^{\frac12}.$$
Then, for each $n$, there exists a $\la_n$--CPS
$(\bar S^n,\bar Q^n)$, where $\bar Q^n$ is the unique martingale measure for $\bar S^n$, and such  that $(\bar Q^n)_{n\geq1}\vartriangleleft\vartriangleright (P)$.
\end{theorem}

It is a direct consequence of Theorem~\ref{Ex1} that there does not exist any form of asymptotic arbitrage with transaction costs $\la_n>2(1-e^{-\gamma_n})$. Indeed, this follows as $(\bar Q^n)_{n\geq1}\vartriangleleft\vartriangleright (P)$ by an application of Theorems~\ref{1}, \ref{2}, \ref{strong}.

In particular, we get as an easy consequence

\begin{corollary}
Fix any $\la>0$. Then, for each $n$, there exists a $\la$--CPS
$(\hat S^n,\hat Q^n)$, where $\hat Q^n$ is the unique martingale measure for $\hat S^n$, and such  that $(\hat Q^n)_{n\geq1}\vartriangleleft\vartriangleright (P)$.
\end{corollary}

\begin{proof}
Choose $N$ such that $\la>2(1-e^{-\gamma_n})$ for all $n\geq N$. Define $(\hat S^n,\hat Q^n)=(\bar S^n, \bar Q^n)$ for $n\geq N$ and $(\hat S^n,\hat Q^n)= (S^n,Q^n)$ for $n<N$.
\end{proof}

In order to prove Theorem~\ref{Ex1} we define the following process $(\tilde{S}^n_t)_{0\leq t\leq T_n}$ which we will use to find the  $\lambda_n$--CPS on the market $n$.

\begin{equation}
\frac{d\tilde S^n_t}{\tilde S^n_t}=dW_t+\frac1{T_n}\tilde H^n_tdt,\label{tilde S^n}
\end{equation}

where
$$\tilde H^n_t=\begin{cases} \frac1{(T_n-t+\alpha_n)^{\frac12}} & \text{for $0\leq t\leq T_n-\alpha_n^{\frac1{T_n^{\ep}}}$}\\
\frac1{(T_n-t+\alpha_n)^{\frac12\left(1-\frac1{T_n^2}\right)}} & \text{for $T_n-\alpha_n^{\frac1{T_n^{\ep}}}< t\leq T_n$}
\end{cases}
$$

We will prove that the following holds.

\begin{lemma}\label{bidask}
Let $\la_n$ as in Theorem~\ref{Ex1}. There exists $N$ such that for each $n\geq N$ there exists $c(\lambda_n)$  such that
$$(1-\lambda_n)S^n_t\leq c(\lambda_n)\tilde S^n_t\leq S^n_t,\quad\quad \text{for all $0\leq t\leq T_n$.}$$
\end{lemma}

Lemma~\ref{bidask} shows that  the process $c(\lambda_n)\tilde S^n$ lies in the bid--ask spread of $S^n$ for transaction costs $\lambda_n$.

Moreover, for each $n$, there exists a unique equivalent martingale measure $\tilde Q^n$ for $\tilde S^n$ given by

$$
\frac{d\tilde Q^n}{dP}=\exp\left(-\frac1{T_n}\int_0^{T_n} \tilde H^n_tdW_t-\frac1{2T^2_n}\int_{0}^{T_n}(\tilde H^n_t)^2dt.
\right)
$$

Indeed, Novikov's condition is satisfied as

\begin{align}
&\frac1{T_n^2}\int_{0}^{T_n}(\tilde H_t^n)^2dt = \frac1{T_n^2}\left(\int_0^{T_n-\alpha_n^{\frac1{T_n^{\ep}}}}\frac1{T_n-t+\alpha_n}dt +
\int_{T_n-\alpha_n^{\frac1{T_n^{\ep}}}}^{T_n}\frac1{(T_n-t+\alpha_n)^{\left(1-\frac1{T_n^2}\right)}}dt\right)\nonumber
\\
&=\frac1{T_n^2}\left(\ln(T_n+\alpha_n)-\ln(\alpha_n^{\frac1{T_n^{\ep}}}+\alpha_n) + T_n^2\left( (\alpha_n^{\frac1{T_n^{\ep}}}+\alpha_n)^{\frac1{T_n^2}}-(\alpha_n)^{\frac1{T_n^2}}\right)\right)\nonumber
\\
&=
\frac1{T_n^2}\left(\ln(T_n+\alpha_n)+T_n^2-\ln(1+e^{-T_n^2(T_n^{\ep}-1)})\right) + e^{-1}(1+e^{-T_n^2(T_n^{\ep}-1)})^{\frac1{T_n^2}}-e^{-T_n^{\ep}}
<\infty\nonumber
\end{align}

(In the calculation we have used that $\alpha_n^{\frac1{T_n^{\ep}}}+\alpha_n=e^{-T_n^2}(1+e^{-T_n^2(T_n^{\ep}-1)})$.)
\newline\newline
But, moreover, we have that Novikov's condition is satisified {\it "in the limit"}, as well, as

\begin{align}\label{Novikov_limit}
&\lim_{n\to\infty}\left(\frac1{T_n^2}\left(\ln(T_n+\alpha_n)+T_n^2-\ln(1+e^{-T_n^2(T_n^{\ep}-1)})\right) + e^{-1}(1+e^{-T_n^2(T_n^{\ep}-1)})^{\frac1{T_n^2}}-e^{-T_n^{\ep}}\right)\\
&=1+e^{-1}<\infty.\nonumber
\end{align}

This gives the following lemma.

\begin{lemma}\label{bicon}
We have that $(\tilde Q^n)_{n\geq1}\vartriangleleft\vartriangleright P$.
\end{lemma}

With Lemma~\ref{bicon} and Lemma~\ref{bidask} the proof of Theorem~\ref{Ex1} follows immediately.

\begin{proof}[Proof of Theorem~\ref{Ex1}]
For $\lambda_n$ as in the Theorem choose the  following $\lambda_n$--CPS $(\bar S^n,\bar Q^n)$ for all $n\geq1$. Take the $N$ and $c(\la_n)$ for $n\geq N$ of Lemma~\ref{bidask}. For $n<N$ let $\bar{S}^n=S^n$ and $\bar{Q}^n=Q^n$. For $n\geq N$ let $\bar{S}^n=c(\lambda_n)\tilde S^n$ and $\bar{Q}^n=\tilde Q^n$. This is the $\la_n$--CPS of Theorem~\ref{Ex1}, everything holds.
\end{proof}

\begin{proof}[Proof of Lemma~\ref{bicon}]

\begin{align}&E_P\left[\left(\frac{d\tilde Q^n}{dP}\right)^2\right]\nonumber\\
&= E_P\left[\exp\left(-\frac2{T_n}\int_0^{T_n} \tilde H^n_tdW_t-\frac1{T_n^2}\int_{0}^{T_n}(\tilde H^n_t)^2dt
\right)\right]\nonumber\\
&=E_P\left[\exp\left(-\frac2{T_n}\int_0^{T_n} \tilde H^n_tdW_t\right)\right]\exp\left(-\frac1{T_n^2}\int_{0}^{T_n}(\tilde H^n_t)^2dt\right)\nonumber\\
&=\exp\left(\frac2{T_n^2}\int_{0}^{T_n}(\tilde H^n_t)^2dt  -\frac1{T_n^2}\int_{0}^{T_n}(\tilde H^n_t)^2dt\right)=\exp\left(\frac1{T_n^2}\int_{0}^{T_n}(\tilde H^n_t)^2dt\right)\nonumber\\
&=\exp\left( \frac1{T_n^2}\left(\ln(T_n+\alpha_n)+T_n^2-\ln(1+e^{-T_n^2(T_n^{\ep}-1)})\right) + e^{-1}(1+e^{-T_n^2(T_n^{\ep}-1)})^{\frac1{T_n^2}}-e^{-T_n^{\ep}}\right)\nonumber\\
&:=\zeta_n.\nonumber
\end{align}

\noindent(In the above we used that $\frac1{T_n}\int_0^{T_n}\tilde{H}^n_tdW_t$ has a normal distribution with mean 0 and variance
$\frac1{T_n^2}\int_0^{T_n}(\tilde{H}^n_t)^2dt$. And  for $X\sim N(0,\sigma^2)$ we have that $E[e^{uX}]=e^{\frac{u^2\sigma^2}2}$.)
\newline\newline
As
$\lim_{n\to\infty}\zeta_n=\exp(1+e^{-1})$
we have that
$E_P\left[\left(\frac{d\tilde Q^n}{dP}\right)^2\right]\leq e^2$,
for $n$ large enough. This implies that $\left(\frac{d\tilde Q^n}{dP}\right)_{n\geq1}$ is uniformly integrable with respect to $P$. Therefore $(\tilde Q^n)\vartriangleleft P$, by Lemma~V.1.10 of \cite{Jac:Shir}. The same argument (using the Brownian motion $W^{\tilde Q^n}=W_t+\frac1{T_n}\int_0^{T_n}\tilde{H}^n_tdt$) shows that $E_{\tilde Q^n}\left[\left(\frac{dP}{d\tilde Q^n}\right)^2\right]\leq e^2$ for $n$ large enough which gives
that
$\left(\frac{dP}{d\tilde Q^n}\right| \tilde Q^n)_{n\geq1}$ is uniformly integrable and therefore $P\vartriangleleft (\tilde Q^n)$, again by  Lemma~V.1.10 of \cite{Jac:Shir}.
\end{proof}

\begin{proof}[Proof of Lemma~\ref{bidask}]
Fix  $2(1-e^{-\gamma_n})<\la_n<1$. Define $\la_n'=\frac{\la_n}{2-\la_n}$, then $\frac{1-\la_n'}{1+\la_n'}=1-\la_n$. For $t\leq T_n-\alpha_n^{\frac1{T_n^{\ep}}}$ we have that $S^n_t=\tilde S^n_t$ and for $T_n-\alpha_n^{\frac1{T_n^{\ep}}}<t\leq T_n$

\begin{equation}\tilde S^n_t=S^n_t\exp\left(\frac1{T_n}\int_{T_n-\alpha_n^{\frac1{T_n^{\ep}}}}^t\frac1{(T_n-u+\alpha_n)^{\frac12(1-\frac1{T_n^2})}}du-
\frac1{T_n}\int_{T_n-\alpha_n^{\frac1{T_n^{\ep}}}}^t\frac1{(T_n-u+\alpha_n)^{\frac12}}du
\right)\nonumber\end{equation}

We will first show that there exists $N$ such that for all $n\geq N$, all $T_n-\alpha_n^{\frac1{T_n^{\ep}}}\leq t\leq T_n$,

\begin{equation}\label{la'}
(1-\la_n')S^n_t\leq \tilde S^n_t\leq (1+\la_n')S^n_t.
\end{equation}

So we have to show that

$$
-\delta_n\leq \frac1{T_n}\int_{T_n-\alpha_n^{\frac1{T_n^{\ep}}}}^t\frac1{(T_n-u+\alpha_n)^{\frac12(1-\frac1{T_n^2})}}du-\frac1{T_n}
\int_{T_n-\alpha_n^{\frac1{T_n^{\ep}}}}^t\frac1{(T_n-u+\alpha_n)^{\frac12}}du\leq \ep_n,
$$

where $\delta_n:=-\ln(1-\la_n')$ and $\ep_n:=\ln(1+\la_n')$. Note that
\begin{equation}\label{rate}
\min(\delta_n,\ep_n)=\ep_n=\ln\big(\frac{2}{2-\la_n}\big).
\end{equation}
By calculating the integrals and rearranging this we see that it is sufficient to show that

\begin{equation}\label{lnlala}
-\ep_n\leq I_n(t) \leq \ep_n,
\end{equation}

for all $T_n-\alpha_n^{\frac1{T_n^{\ep}}}\leq t\leq T_n$ where

\begin{gather} I_n(t)=\nonumber\\
-\frac2{T_n}(\alpha_n^{\frac1{T_n^{\ep}}}+\alpha_n)^{\frac12}\left( 1-\frac{(\alpha_n^{\frac1{T_n^{\ep}}}+\alpha_n)^{\frac1{2T_n^2}}}{1+\frac1{T_n^2}} \right)
+\frac2{T_n}(T_n-t+\alpha_n)^{\frac12}\left( 1-\frac{(T_n-t+\alpha_n)^{\frac1{2T_n^2}}}{1+\frac1{T_n^2}} \right)
\nonumber
\end{gather}

We can choose $N$ such that for $n\geq N$ we have that $1-\frac{(\alpha_n^{\frac1{T_n^{\ep}}}+\alpha_n)^{\frac1{2T_n^2}}}{1+\frac1{T_n}}\geq 0$.
Indeed, this holds as $(\alpha_n^{\frac1{T_n^{\ep}}}+\alpha_n)^{\frac1{2T_n^2}}=e^{-\frac12}
(1+e^{-T_n^2(T_n^{\ep}-1)})^{\frac1{2T_n^2}}\to e^{-\frac12}<0.6$ for $n\to\infty$.
Therefore, $0\leq 1-\frac{(T_n-t+\alpha_n)^{\frac1{2T_n^2}}}{1+\frac1{T_n^2}}\leq 1$ for all $T_n-\alpha_n^{\frac1{T_n^{\ep}}}\leq t\leq T_n$. This implies that, for all $n\geq N$ and all
$T_n-\alpha_n^{\frac1{T_n^{\ep}}}\leq t\leq T_n$,

\begin{equation}\label{bound}
|I_n(t)|\leq \frac2{T_n}(\alpha_n^{\frac1{T_n^{\ep}}}+\alpha_n)^{\frac12}=\frac2{T_n}e^{-\frac{T_n^2}2}(1+e^{-T_n^2(T_n^{\ep}-1)})
^{\frac12}=\gamma_n.
\end{equation}
As for $\la_n>2(1-e^{-\gamma_n})$, we have that
$$\ep_n=\ln\big(\frac{2}{2-\la_n}\big)>\gamma_n$$
we get
by (\ref{bound}) that $|I_n(t)|< \ep_n$ for all $n\geq N$ and all
$T_n-\alpha_n^{\frac1{T_n^{\ep}}}\leq t\leq T_n$ which implies (\ref{lnlala}) and therefore (\ref{la'}).

Clearly we have that

$$\frac{1-\la_n'}{1+\la_n'}S^n_t\leq \frac1{1+\la_n'}\tilde S^n_t\leq \frac{1+\la_n'}{1+\la_n'}S^n_t.$$

So, if we define $c(\la_n)=\frac1{1+\la_n'}=\frac{2-\la_n}{2}$ and recall that $\frac{1-\la_n'}{1+\la_n'}= 1-\la_n$ we get

$$(1-\la_n)S^n_t\leq c(\la_n) \tilde S^n_t\leq S^n_t.$$

And for $0\leq t\leq T_n-\alpha_n^{\frac1{T_n^{\ep}}}$, trivially,

$$(1-\la_n)S^n_t=\frac{1-\la_n'}{1+\la_n'}S^n_t\leq \frac{1}{1+\la_n'}S^n_t=c(\la_n)\tilde S^n_t\leq S^n_t.$$
\end{proof}

\begin{remark}
It would be interesting to determine the fastest possible rate of convergence of $\la_n\to0$ such that there does not exists an asymptotic arbitrage, and such that whenever $\la_n\to0$ faster than that there exists a strong asymptotic arbitrage. We leave this open for further research.
\end{remark}

\section{Appendix }\label{app}
\subsection{Some results for a market with small transaction costs:}
We consider one market $n$ and omit the superscript $n$. Each of the following results will hold on each market $n$.
The following theorem can be found in \cite{G:R:S:2008} (and in \cite{Sch}).

\begin{theorem}\label{NA-L} The following two statements are equivalent.
\begin{enumerate}
\item
$\mathcal{M}(\la)\ne\emptyset$ for each $\la>0$.
\item NA($\la$) is satisfied, for each $\la>0$.
\end{enumerate}
\end{theorem}

\begin{lemma}\label{blum} Let $\phi=(\phi^0,\phi^1)$ be a self--financing trading strategy.  Then
$$V_t^{0,\widetilde{S}}(\phi)=\phi_t^0+\phi_t^1\widetilde{S}_t\geq V^{\lambda, S}_t(\phi)\quad\quad\text{for all $0\leq t\leq T$},$$
for any  $\widetilde{S}$ with $(1-\la)S_t\leq \widetilde{S}_t\leq S_t$ a.s., $0\leq t\leq T$.
\end{lemma}

\begin{proof}
Immediate by the definition of $\widetilde{S}$.
\end{proof}

The following lemma is proved in an unpublished work of W.~Schachermayer \cite{Sch}. For the convenience of the reader we present the easy proof.

\begin{lemma}\label{grr}
Suppose that $\mathcal{M}(\la)\neq\emptyset$.  Let $\phi=(\phi^0,\phi^1)$ be a self--financing admissible trading strategy for $S$ with transaction costs $\la$. Let $(\widetilde{S},Q)$ be a $\la$--consistent  pricing system. Then the process
$$V_t^{0,\widetilde{S}}(\phi)=\phi_t^0+\phi_t^1\widetilde{S}_t,\quad,0\leq t\leq T$$
is a $Q$--local martingale which is uniformly bounded from below and therefore a $Q$--supermartingale.
\end{lemma}

\begin{proof}
$(\phi^1_t)_{0\leq t\leq T}$  is a finite--variation process and $\widetilde{S}$ is continuous, therefore we can apply the product rule and get
\begin{equation}\label{Vtilde}
dV_t^{0,\widetilde{S}}(\phi)=d\phi_t^0+\widetilde{S}_td\phi_t^1+\phi_t^1d\widetilde{S}_t.
\end{equation}
By (\ref{selff}) we have that
\begin{equation}\label{selffagain}
d\phi_t^{0}\leq (1-\la)S_td\phi_t^{1,\downarrow}- S_td\phi_t^{1,\uparrow}, \quad 0\leq t\leq T.
\end{equation}
As $\widetilde{S}_t$ takes its values in the bid--ask--spread $[(1-\la)S_t,S_t]$ we see from (\ref{selffagain}) that the process
$(\int_0^t(d\phi_t^0+\widetilde{S}_td\phi_t^1)_{0\leq t\leq T}$ is decreasing. By the admissibility, see (\ref{valuepr}), and by Lemma~\ref{blum} the process $\widetilde{V}$ is uniformly bounded from below. Therefore the stochastic integral $\int_0^t\phi_t^1d\widetilde{S}_t$ is uniformly bounded from below as well and hence a  $Q$--local martingale which is uniformly bounded from below, hence a $Q$--supermartingale. Therefore  $\widetilde{V}$ is a $Q$--supermartingale as well.
\end{proof}

There are a number of superreplication theorems in the presence of transaction costs. For a general version that implies the theorem below, see, e.g., \cite{C:S}. In
the present form it was communicated to us in a personal communication with W.~Schachermayer \cite{Sch}.

\begin{theorem}[Superreplication Theorem]\label{super}
Let $f$ be a random variable with $f\geq -M$ for some $M>0$. Then the following are equivalent.
\begin{enumerate}
\item There is a self--financing admissible trading strategy $\phi=(\phi^0,\phi^1)$ starting at $(\phi^0_{0-},\phi^1_{0-})=(0,0)$ such that the terminal value is equal to $\phi_T=(f,0)$. That is
    $$V^{\la,S}_T(\phi)=f.$$
\item $E_Q[f]\leq 0$ for every  $Q\in\mathcal{M}(\la)$. \end{enumerate}
\end{theorem}

The following proposition was communicated to us by W.~Schachermayer \cite{Sch} and will appear in a future work of him together with an example that shows that the assumption that $\mathcal{M}(\la')\ne\emptyset$ for each $\la'>0$ is indeed necessary. This example shows that even if
one excludes arbitrage with transaction costs $\la$, it is possible to construct a value
process which  is bounded from below by $-1$ at terminal time, but it is strictly smaller than $-1$ with positive
probability at a prior time point.

\begin{proposition}\label{schach}
Assume that $\mathcal{M}(\la')\ne\emptyset$ for each $\la'>0$. Assume that $\phi=(\phi^0,\phi^1)$  is
a self--financing admissible  trading  strategy under transaction costs $\la$ for $S$. Suppose that there is $M>0$  such that
$$V^{\la,S}_T(\phi)\geq-M.$$
 Then we have that
$$V^{\la,S}_{\tau}(\phi)\geq -M,\quad\quad\text{a.s.}$$
for each stopping time $0\leq \tau\leq T$ (in particular for deterministic time points $t\leq T$).
\end{proposition}

\subsection{Quantitative versions of the Halmos-Savage Theorem:}
The following two quantitative version of the Halmos-Savage Theorem were proved in \cite{K:Sch:1996}.

\begin{proposition}[Quantitative version of the Halmos-Savage
Theorem]\label{HS1}
For fixed $\ep>0$ and $\delta>0$ the following statement is true:
let $M$ be  a convex set of $\P$-absolutely continuous probability measures
such that for all sets $A\in\mathcal{F}$ with
$P(A)>\ep$ there is $Q\in M$ with $Q(A)>\delta$. Then
there is $Q_0\in M$ such that for all $A\in\mathcal{F}$ with $P(A)>4\ep$ we have that $Q_0(A)>\frac {\ep^2\delta} {2}$.
\end{proposition}

\begin{proposition}[Dual quantitative version of the
Halmos-Savage Theorem]\label{HS2}
For fixed $\ep>0$ and $\de>0$ the following statement is true: let
$M$ be a convex set of $\P$-absolutely continuous probability
measures  such that for all sets $A\in\mathcal{F}$ with
$P(A)<\de$ there is $Q\in M$ with $Q(A)<\ep$. Then
there is $Q_0\in M$ such that for all $A\in\mathcal{F}$ with $P(A)<2\ep\de$ we have that $Q_0(A)<8\ep$.
\end{proposition}

%
%
%
%

\bibliographystyle{plain}

\end{document}